%% file: main.tex
\documentclass[10pt,conference]{IEEEtran}
\IEEEoverridecommandlockouts
\usepackage{authblk}
\usepackage{cite}
\usepackage{amsmath,amssymb,amsfonts,dsfont}
\usepackage{algorithm}
\usepackage[noend]{algorithmic}
\usepackage{graphicx}
\usepackage{textcomp}
\usepackage{xcolor}
\usepackage{bm}
\usepackage{booktabs}
\usepackage{multirow}
\usepackage{lipsum}
\usepackage{url}
\usepackage{amsthm}
\usepackage{subcaption}

\ifCLASSOPTIONcompsoc
\usepackage[caption=false, font=normalsize, labelfont=sf, textfont=sf]{subfig}
\else
\usepackage[caption=false, font=footnotesize]{subfig}

\newtheorem{theorem}{Theorem}
\newtheorem{proposition}{Proposition}

\newtheorem{lemma}{Lemma}

\usepackage{geometry}
\begin{document}

\title
{
Estimating Rate-Distortion Functions Using the Energy-Based Model
\\
\thanks{The first two authors contributed equally to this work and $\dag$ marks the corresponding author. The authors would like to thank Prof. Hao Wu for valuable discussions and helpful advice. This work was supported by the National Natural Science Foundation of China (Grant Nos. 12271289 and 62231022).}
}

\author[1]{Shitong Wu}
\author[1]{Sicheng Xu}
\author[1]{Lingyi Chen}
\author[2]{Huihui Wu}
\author[3$\dag$]{Wenyi Zhang}

\affil[1]{Department of Mathematical Sciences, Tsinghua University, Beijing 100084, China}
\affil[2]{Zhejiang Key Laboratory of Industrial Intelligence and Digital Twin, 
\authorcr Eastern Institute of Technology, Ningbo, Zhejiang 315200, P.R. China
} 
\affil[3]{Department of Electronic Engineering and Information Science, 
\authorcr University of Science and Technology of China, Hefei, Anhui 230027, China 
\authorcr Email: wenyizha@ustc.edu.cn 
\vspace{-10pt}
}

\maketitle

\begin{abstract}
% RD 重要
% 计算RD的意义
% BA是经典算法，问题
% 神经网络
The rate-distortion (RD) theory is one of the key concepts in information theory, providing theoretical limits for compression performance and guiding the source coding design, with both theoretical and practical significance. 
% 逻辑问题，however是反对率失真的定义
% BA是计算RD的经典方法，遇到computational困难在高维情况
% 高维的情况换成，分布函数可以被采样？ i 数据集；ii 连续
The Blahut-Arimoto (BA) algorithm, as a classical algorithm to compute RD functions, encounters computational challenges when applied to high-dimensional scenarios.
In recent years, many neural methods have attempted to compute high-dimensional RD problems from the perspective of implicit generative models. 
Nevertheless, these approaches often neglect the reconstruction of the optimal conditional distribution or rely on unreasonable prior assumptions. 
In face of these issues, we propose an innovative energy-based modeling framework that leverages the connection between the RD dual form and the free energy in statistical physics, achieving effective reconstruction of the optimal conditional distribution. 
The proposed algorithm requires training only a single neural network and circumvents the challenge of computing the normalization factor in energy-based models using the Markov chain Monte Carlo (MCMC) sampling. 
% 思考分布形式或者高维，注意是连续非离散
Experimental results demonstrate the significant effectiveness of the proposed algorithm in estimating high-dimensional RD functions and reconstructing the optimal conditional distribution. 
\end{abstract}
\vspace{-2pt}

\section{Introduction}
%介绍RD背景，计算RD函数重要
The rate-distortion (RD) theory, first introduced by Shannon in \cite{shannon1948mathematical, shannon1959coding}, is one of the fundamental concepts in information theory.
It  characterizes the minimum amount of information (rate) required to represent a source while maintaining a specified level of distortion.
By establishing theoretical limits on compression performance, the RD theory not only guides the development of practical source coding schemes \cite{balle2017end,balle2020nonlinear} but also enables the evaluation of their capabilities.
Therefore, solving the RD problem is both a theoretically significant and practically essential endeavor.

%算法：离散BA，AS，但是难以解决高维情形，连续sandwich，NERD，WGD
It is known that the classical Blahut-Arimoto (BA) algorithm \cite{blahut1972computation} can efficiently solve the RD problem with discrete source distributions.
%
%However, for continuous cases, applying the algorithm requires discretization, which leads to an exponential increase in computational complexity.
% 这里多花一个限定词说一下（连续，非离散，高维）
However, when applied to continuous cases, the algorithm necessitates discretization, which results in an exponential increase in computational complexity\cite{10619450,lei2022neural}. 
Therefore, the  BA algorithm is not always applicable for continuous and high-dimensional scenarios.

In recent years, several novel approaches have been proposed to solve the high-dimensional RD problem.
For instance, the authors of \cite{yang2022towards} provide sandwich bounds for the RD function, where the upper bound is obtained utilizing the $\beta$-VAE \cite{higgins2017beta} architecture and has been validated to be tight for high-dimensional scenarios.
Meanwhile, \cite{lei2022neural} represents the reproduction distribution by a mapping and optimizes the dual function via neural parameterization.
Another work \cite{yang2024estimating} 
introduces a neural-network-free algorithm, using Wasserstein gradient descent in optimal transport to optimize the distribution.

%能量模型的优势：NERD，WGD只能算边缘分布，得不到条件分布，对编码没有启发作用；sandwich的条件概率基于松弛，算得不准
However, while the dual function framework utilized in \cite{lei2022neural,yang2024estimating} allows for a more compact formulation, a limitation of these works is that they consider only the optimization of the marginal distribution, neglecting the reconstruction of the conditional distribution. 
% 说不好，去解释；说有limit是什么
%This is noteworthy since the conditional distribution not only induces the optimal solution but also provides valuable insights for compression. 
%
{This is noteworthy since the conditional distribution is not only inherent to the formulation of the RD problem but also provides valuable insights for compression. }
Although the method in \cite{yang2022towards} can obtain both the marginal distribution and the conditional distribution, it relies on an assumption about the prior of the reproduction data, which is not always reasonable. 
As a result, effective and accurate reconstruction of the conditional distribution still remains unaddressed.

%以上是隐式生成模型，随着扩散模型（ddpm）的发展，转向score（宋）函数的研究很热门，这都是基于对分布的能量建模。因此能量建模得到广泛的关注

%值得注意的是最近有一片文章从能量角度算，基于重建分布和条件概率最优性条件的关系，发现最优条件分布和最优边缘分布可以通过类似 Gibbs分布的形式联系起来；但也做了松弛，用两张网络，需要内置的diffusion model，计算复杂度高
%
Existing neural methods \cite{lei2022neural,yang2022towards} estimating the RD function primarily model probability distribution through the lens of implicit generative models. 
% 这里是一个递进关系，不是转折关系，不能用however
% 随着机器学习技术的进一步发展，扩散模型取得了很效果，能量建模取得关注
% energy-based models (EBM) 改成能量形式的建模
% 
% 能量建模时一种建模技术，得关注，由于扩散模型展现出的强大效果和影响
However, with the recent breakthroughs in diffusion models \cite{ho2020denoising} and score-based models \cite{DBLP:conf/iclr/0011SKKEP21}, the energy-based modeling framework has emerged as a compelling alternative to probabilistic modeling.
% 
% However, the energy-based modeling framework has emerged as a compelling alternative to probabilistic modeling, largely driven by recent breakthroughs in diffusion models \cite{ho2020denoising} and score-based models \cite{DBLP:conf/iclr/0011SKKEP21}.
% 
% However, with the development of diffusion models \cite{ho2020denoising} and score-based models \cite{DBLP:conf/iclr/0011SKKEP21}, the energy-based models (EBM), as an alternative approach to probabilistic modeling, have gained significant attention.
% 
%
Notably, a recent work \cite{li2024rate} approaches the RD problem from an energy perspective, finding that the optimal conditional distribution and the optimal marginal distribution are linked via the Boltzmann distribution. 
The authors of \cite{li2024rate} innovatively employ a single network to represent both distributions. 
Nevertheless, this method does not fully leverage the properties of the RD function.
%
%It additionally introduces a relaxation to the original RD Lagrangian by approximating mutual information \cite{DBLP:conf/icml/BelghaziBROBHC18}, increasing the estimation error and computational complexity.
It approximates the mutual information in the original RD Lagrangian with an additional neural network \cite{DBLP:conf/icml/BelghaziBROBHC18}, increasing the estimation error and computational complexity.
%将能量模型代入RD lagrangian形式很困难，注意到它的对偶形式恰好可以类比于物理学中的自由能，所以我们接下来从对偶形式角度进行建模

%我们提出了一个更加简洁的基于能量模型的RD估计算法。不同于前一篇文章，我们结合了对偶函数和自由能的联系，以及边缘分布与条件分布的能量联系。我们将能量形式的边缘分布代入对偶函数当中，得到单变量的优化模型。为解决能量模型存在不可计算的归一化因子这一问题，我们计算优化目标的梯度，并利用mcmc对其进行估计，从而得到了对能量函数的训练算法。此外，我们发现算法与极大似然的能量模型近似，通过分析得到其本质上的联系。实验结果验证了方法的有效性。
In this paper, we propose a more concise algorithm for estimating RD functions using the energy-based model.
%
% In contrast to \cite{li2024rate}, we take advantage of both the simplicity of the dual function and the energy-based relationship between the optimal marginal distribution and the optimal conditional distribution.
%
Different from \cite{li2024rate}, we recognize that the dual form of the RD function corresponds to the free energy in statistical physics \cite{landau1980statistical}, based on which our approach follows.
% 得到优化目标，这个目标作为loss函数
% solved as 
Specifically, by substituting the marginal distribution, as represented by the energy-based model, into the dual function, we obtain the optimization objective. 
% 计算倒数困难，？？？？
% 困难是由于某个难以得到的项导致，（没有解释）
% 困难是由计算某个项导致的 （解释计算）
% 使用song求梯度，正则项就没了；或者是变简单了
% 为什么就是计算梯度呢？ 
% avoid 不是很学术？
% To avoid the challenge posed by the intractable normalization factor in energy-based models, we compute the gradient of the objective instead of direct evaluation of the objective function. 
% % 
% and estimate it using the Markov chain Monte Carlo (MCMC) \cite{parisi1981correlation,grenander1994representations}, thereby developing a training algorithm for the energy function. 
%
% To avoid the challenge posed by the intractable normalization factor in energy-based models, we circumvent direct evaluation of the objective function. 
%
To avoid the challenge posed by the intractable normalization factor in energy-based models, we compute the gradient of the objective and estimate it using Langevin Markov chain Monte Carlo (Langevin MCMC) methods \cite{parisi1981correlation,grenander1994representations}, which enable practical training of the energy function through gradient-based optimization. 
%
%The proposed algorithm requires training only a single neural network and capitalizes on the energy-based relationship between the optimal marginal distribution and the optimal conditional distribution. 
%
{Furthermore, we analyze the asymptotic property of the proposed estimation algorithm and elucidate its connection with classical energy-based models.}
%Furthermore, we demonstrate that the proposed algorithm is formally analogous to the classical energy-based models and establish their intrinsic connection via theoretical analysis. 
%
Finally, numerical experiments show the effectiveness of the proposed algorithm. 
% 补citation

\vspace{-2pt}
\section{Preliminaries}
\subsection{Energy-Based Models}

The energy-based model (EBM) \cite{lecun2006tutorial} is a framework for modeling the probability density of variables with an energy function.
Specifically, let $x$ denote a random variable in the continuous space $\mathcal{X}$ and $E_\theta(x):\mathcal{X}\rightarrow \mathcal{R}$ represent an energy function parameterized by a neural network.
The probability density of $x$ is given by the Boltzmann distribution:
\vspace{-5pt}
\begin{equation*}
    p_\theta(x)=\frac{e^{-E_\theta(x)}}{Z_\theta},
    \vspace{-7pt}
\end{equation*}
where $Z_\theta=\int_\mathcal{X} e^{-E_\theta(x)}\mathrm{d}x$ is the normalization factor.

The key advantage of energy-based models lies in their ability to efficiently compute the gradient of the log-probability (also known as the score) with respect to $x$:
\vspace{-5pt}
\begin{equation*}
    \nabla_x \log p_\theta(x)=-\nabla_x E_\theta(x).
    \vspace{-5pt}
\end{equation*}
Then a sample of $x$ can be obtained through the Langevin MCMC \cite{parisi1981correlation,grenander1994representations}. 
The update rule is given by
\vspace{-3pt}
\begin{equation*}
    x_{k}=x_{k-1}-\frac{\epsilon^2}{2}\nabla_x E_\theta(x_{k-1})+\epsilon z_k,
    \vspace{-2pt}
\end{equation*}
where $\epsilon$ is the step size, and $z_k\sim \mathcal{N}(0,1)$ is a noise term. 
As $\epsilon\rightarrow0$ and the number of steps $K\rightarrow+\infty$, the distribution of $x_K$ is guaranteed to converge to $p_\theta(x)$.

The energy function $E_\theta(x)$ can be trained using maximum likelihood estimation (MLE)\cite{song2021train, kumar2019maximum}.
Specifically, let $p_{\text{data}}(x)$ be the real distribution of the data, and then the target is to minimize the expectation of negative log-likelihood function
\vspace{-2pt}
\begin{equation*}
    \mathbb{E}_{x\sim p_{\text{data}}(x)}\left[-\log p_\theta(x) \right].
    \vspace{-2pt}
\end{equation*}
The function itself is hard to estimate as $Z_\theta$ is intractable. 
Nevertheless, its gradient can be simplified by 
\begin{equation*}
\begin{aligned}
    &\nabla_{\theta} \mathbb{E}_{x\sim p_{\text{data}}(x)}\left[-\log p_{\theta}(x)\right]\\=&-\mathbb{E}_{x\sim p_{\text{data}}(x)}\left[-\nabla_{\theta} E_{\theta}(x)\right]+\nabla_{\theta} \log Z_{\theta} \\
    =& \mathbb{E}_{x\sim p_{\text{data}}(x)}\left[\nabla_{\theta} E_{\theta}(x)\right]-
    \mathbb{E}_{x\sim p_{\theta}(x)}\left[\nabla_{\theta} E_{\theta}(x) \right].
\end{aligned}
\end{equation*}

Thus, it can be estimated using Monte Carlo methods to perform gradient descent.

\subsection{Rate-Distortion Theory}
%
%Rate-distortion (RD) theory \cite{shannon1948mathematical,shannon1959coding} analyzes the fundamental tradeoff between the rate rate required to represent samples from a data source $X \sim P_X$, and the expected distortion between the original data and the reconstructed one.
%
The RD function characterizes the trade-off between the	compression rate and reconstruction fidelity in lossy source coding.
Given a memoryless source $X\sim P_X$ supported on the alphbet $\mathcal{X}$ with reproduction $Y\sim \mathcal{Y}$, the RD function is defined as
\begin{equation}
    R(D):=\inf _{P_{Y \mid X}: \mathbb{E}_{P_{X, Y}}[\rho(X, Y)] \leq D} I(X ; Y),
    \label{RD}
\end{equation}
where $\rho:\mathcal{X}\times \mathcal{Y}\rightarrow [0,+\infty)$ is the distortion measure.

A classical numerical approach to computing the RD function is the Blahut-Arimoto (BA) algorithm \cite{blahut1972computation}.
This method addresses the RD problem in discrete scenarios by iteratively optimizing the following two variables:
\begin{equation*}
\left\{
\begin{aligned}
p(y|x) & =\frac{q(y) e^{-\beta \rho(x, y)}}{\int_{\mathcal{Y}}q(y) e^{-\beta \rho(x, y)}\mathrm{d}y} \\
q(y) & =\int_{\mathcal{X}}  p(x) p(y|x)\mathrm{d}x
\end{aligned}\right.,
\end{equation*}
where $\beta$ is a fixed Lagrangian multiplier.
The algorithm works well for low-dimensional cases but becomes ineffective or even infeasible as the dimension increases\cite{lei2022neural}.

By  performing variational calculation on the Lagrangian of the RD function\cite{rose1994mapping}, the problem can be transformed into minimizing
\vspace{-2pt}
\begin{equation}
    F(q)=- \int_\mathcal{X} \mathrm{d} x p(x) \log \int_\mathcal{Y} \mathrm{d}y q(y) e^{-\beta \rho(x, y)}.
    \label{vf}
\end{equation}
This formulation is flexible and has led to the development of neural network algorithms specifically designed for high-dimensional cases\cite{lei2022neural}.

\subsection{Connection between Rate-Distortion Problem and Energy-Based Models}
A key connection between energy-based models and the rate-distortion problem lies in the observation that the optimality condition of the RD problem, i.e.,
\begin{equation}
    p(y|x) =\frac{q(y) e^{-\beta \rho(x, y)}}{\int_{\mathcal{Y}}q(y) e^{-\beta \rho(x, y)}\mathrm{d}y},
    \label{pyx}
\end{equation}
also takes the form of a Boltzmann distribution.

The authors of \cite{li2024rate} demonstrated that it is feasible to use a single neural network to represent both distributions $q(y)$ and $p(y|x)$, as detailed in the following lemma.

\begin{lemma}
    In the RD problem, if the optimal marginal distribution $q(y)$ is represented by the energy function $E_\theta(y)$, the conditional distribution $p(y|x)$ can be 
    represented by $E^\prime_\theta(x,y)=E_\theta(y)+\beta \rho(x,y)$, i.e.,
    \begin{equation*}
        \begin{aligned}
            &q_\theta(y)=\frac{e^{-E_\theta(y)}}{\int_{\mathcal{Y}}e^{-E_\theta(y)}\mathrm{d}y},\\
            &p_\theta(y|x)=\frac{e^{-E^\prime_\theta(x,y)}}{\int_{\mathcal{Y}}e^{-E^\prime_\theta(x,y)}\mathrm{d}y}.\\
        \end{aligned}
    \end{equation*}
\end{lemma}

However, the subsequent relaxation of mutual information in \cite{li2024rate} increases the complexity of the algorithm.
In this paper, we take Lemma 1 as a start point, and then investigate more intrinsic connections between the RD function and energy-based models, subsequently devising a more concise algorithm.
%Li obtained an estimation algorithm for the RD function through this discovery and by relaxing the mutual information\cite{belghazi2018mine}.
%
%In this work, we combine this lemma with the variational form \eqref{vf} and conduct further analysis to derive a more concise algorithm.

\section{Energy-Based Model of Rate-Distortion Function}
In this section, we introduce the energy-based rate-distortion (EBRD) model. 
We present the derivation of the algorithm and discuss its relationship with the classical EBMs.

\subsection{Energy-Based Rate-Distortion Model}

Noting that the dual function \eqref{vf} of RD corresponds to the free energy in statistical physics \cite{landau1980statistical}, we adopt this duality as our optimization objective.
By representing the distribution $q(y)$ in \eqref{vf} using an energy-based model $E_\theta(y)$, the loss function to be minimized can be obtained as:
%Consider the functional \eqref{vf} with the distribution $q(y)$ modeled by the energy function $E_\theta(y)$.
%
%Then the objective that needs to be minimized is given by:
\begin{equation*}
    \begin{aligned}
\mathcal{L}(\theta) & =-\mathbb{E}_{P_X}\left[\log \mathbb{E}_{Q_Y^\theta}\left[e^{-\beta \rho(X, Y)}\right]\right] \\
& =-\mathbb{E}_{P_X}\left[\log \int_Y e^{-E_\theta^{\prime}(X, y)} \mathrm{d} y-\log \int_Y e^{-E_\theta(y)} \mathrm{d} y\right] .
\end{aligned}
\end{equation*}
Similar to MLE, there is an intractable term 
$\int_{\mathcal{Y}}e^{-E_\theta(y)} \mathrm{d} y$
 in $\mathcal{L}(\theta)$, which makes direct computation infeasible.
 Therefore, our goal is to compute its gradient with respect to $\theta$, as described in Proposition \ref{g}.
 \begin{proposition}
 Let $P^\theta_{Y |X}$ be the conditional distribution represented by $E^\prime_\theta(x,y)=E_\theta(y)+\beta \rho(x,y)$, and $P^\theta_Y$ be the marginal distribution of the joint distribution $P_XP^\theta_{Y|X}$ over the space $\mathcal{Y}$.
 Then the gradient of  the objective function  $\mathcal{L}(\theta)$ with respect to the parameter $\theta$ is:
 \begin{equation}
     \nabla_\theta \mathcal{L}(\theta)\!=\!\mathbb{E}_{P^\theta_Y}\!\left[\nabla_\theta E_\theta(Y)\right]\!-\!\mathbb{E}_{Q^\theta_Y}\!\left[\nabla_\theta E_\theta(Y)\right].
     \label{gradient}
 \end{equation}
     \label{g}
 \end{proposition}
\vspace{-15pt}
 \begin{proof}
     First, we expand the expression for the gradient of 
$\mathcal{L}(\theta)$ with respect to $\theta$:
\begin{equation*}
    \begin{aligned}
        \nabla_\theta \mathcal{L}(\theta)&\!=\!-\nabla_\theta \mathbb{E}_{P_X}\!\!\left[\log \!\int_\mathcal{Y} e^{-E_\theta^{\prime}(X, y)} \mathrm{d} y\!-\!\log\! \int_\mathcal{Y} e^{-E_\theta(y)} \mathrm{d} y\right]\\
        &\!=\!-\mathbb{E}_{P_X}\!\!\left[\nabla_\theta\!\log\!\! \int_\mathcal{Y}\!e^{-E_\theta^{\prime}(X, y)} \mathrm{d} y\!-\!\nabla_\theta\!\log\!\! \int_\mathcal{Y}\! e^{-E_\theta(y)} \mathrm{d} y\right]\!.
    \end{aligned}
\end{equation*}
For the first term, applying the chain rule of differentiation, we obtain
\vspace{-3pt}
\begin{equation*}
\begin{aligned}
    \nabla_\theta\log \int_\mathcal{Y} e^{-E_\theta^{\prime}(X, y)}\mathrm{d} y&=\frac{\int_\mathcal{Y} e^{-E_\theta^{\prime}(X, y)}(-\nabla_\theta E_\theta^{\prime}(X, y))\mathrm{d} y}{\int_{\mathcal{Y}} e^{-E_\theta^{\prime}(X, y)}\mathrm{d} y}\\
    &=\int_\mathcal{Y} p_\theta(y|x)(-\nabla_\theta E_\theta^{\prime}(X, y))\mathrm{d} y,\\
\end{aligned}
\end{equation*}
where $\nabla_\theta E_\theta^{\prime}(X, y)=\nabla_\theta E_\theta(y)$ due to $E^\prime_\theta(x,y)=E_\theta(y)+\beta \rho(x,y)$.
Thus,
\begin{equation*}
\begin{aligned}
    \nabla_\theta\log \int_\mathcal{Y} e^{-E_\theta^{\prime}(X, y)}\mathrm{d} y&=\int_\mathcal{Y} p_\theta(y|x)(-\nabla_\theta E_\theta(y))\mathrm{d} y\\
    &=-\mathbb{E}_{P^\theta_{Y |X}}\left[\nabla_\theta E_\theta(Y)\right].
\end{aligned}
\end{equation*}

The second term can also be derived in a similar manner:
\begin{equation*}
\begin{aligned}
    \nabla_\theta\log \int_\mathcal{Y} e^{-E_\theta( y)}\mathrm{d} y&=\frac{\int_\mathcal{Y} e^{-E_\theta(y)}(-\nabla_\theta E_\theta(y))\mathrm{d} y}{\int_{\mathcal{Y}} e^{-E_\theta(y)}\mathrm{d} y}\\
    &=\int_\mathcal{Y} q_\theta(y)(-\nabla_\theta E_\theta(y))\mathrm{d} y\\
    &=-\mathbb{E}_{Q^\theta_Y}\left[\nabla_\theta E_\theta(Y)\right].\\
\end{aligned}
\end{equation*}

By substituting the simplified two terms,  the gradient can be further simplified as
\begin{equation*}
\begin{aligned}
    \nabla_\theta \mathcal{L}(\theta)&\!=\!\mathbb{E}_{P_X}\!\left[\mathbb{E}_{P^\theta_{Y |X}}\!\left[\nabla_\theta E_\theta(Y)\right]\!-\!\mathbb{E}_{Q^\theta_Y}\!\left[\nabla_\theta E_\theta(Y)\right]\right]\\
    &\!=\!\mathbb{E}_{P^\theta_Y}\!\left[\nabla_\theta E_\theta(Y)\right]\!-\!\mathbb{E}_{Q^\theta_Y}\!\left[\nabla_\theta E_\theta(Y)\right].
\end{aligned} 
\end{equation*}
Thus, the final formula \eqref{gradient} is obtained.
 \end{proof}

 In practice, one can approximate the expectation term in \eqref{gradient} using Langevin MCMC.
 The specific sampling steps are described as follows:
 \begin{enumerate}
     \item Sample $x_1,x_2,\dots,x_N$ from $P_X$, i.e., derive a random batch from the source dataset.
     \item For each $x_i$, sample $y_i$ from the conditional distribution $P^\theta_{Y|X}$ via Langevin MCMC.
     \item Sample $y^\prime_1,y^\prime_2,\dots,y^\prime_N$ from $Q^\theta_Y$ via Langevin MCMC.
 \end{enumerate}

 Therefore, the formula \eqref{gradient} can be approximated as:
 \begin{equation*}
     \nabla_\theta \mathcal{L}(\theta)\!\approx \frac{1}{N} \sum_{i=1}^{N}\nabla_\theta E_\theta(y_i)-\frac{1}{N} \sum_{j=1}^{N}\nabla_\theta E_\theta(y^\prime_j).
 \end{equation*}

To this end, we have derived a computable approximate form for the gradient of the loss function $\mathcal{L}_{\theta}$.
Therefore, we can use gradient descent to train the energy function $E_\theta(y)$.
The detailed training procedure is summarized in Algorithm \ref{alg}.
{Additionally, we describe in Appendix \ref{app:A} how to estimate the $(R, D)$ pairs using the well-trained energy function $E_\theta$ and establish theoretical guarantees for the proposed estimation method.}

\begin{algorithm}[t]
\renewcommand{\algorithmicrequire}{\textbf{Input:}}
\renewcommand{\algorithmicensure}{\textbf{Output:}}
\renewcommand{\algorithmicreturn}{\textbf{Return:}}
\caption{Training Process of EBRD Model}
\begin{algorithmic}
    \REQUIRE Source distribution $P_X$; distortion measure $\rho(x,y)$; number of steps $K$, and step size $\epsilon$ for Langevin MCMC; sample batch size $N$; learning rate $\eta$ and hyper-parameter $\beta$.
    \ENSURE Parameter $\theta$ of the energy function $E_\theta(y)$.
    \WHILE{$\theta$ does not converge}
        \STATE Sample $x_1,x_2,\ldots,x_N\sim P_X$
        \STATE Sample  $y^{(0)}_1,y^{(0)}_2,\ldots,y^{(0)}_N\sim \mathcal{N}(0,I)$
        \FOR{$k=1,2,\ldots,K$}
        \STATE  Obtain $y^{(k)}_i,i=1,2,\ldots,N$ through Langevin step:
        \begin{equation*}
        y^{(k)}_i=y^{(k-1)}_i-\frac{\epsilon^2}{2}\nabla_y E^\prime_\theta(x_i,y^{(k-1)}_i)+\epsilon z^{(k)}_i,
        \end{equation*}
        where $z^{(k)}_i\sim\mathcal{N}(0,I)$.
        \ENDFOR
        \STATE Sample $y^{\prime(0)}_1,y^{\prime(0)}_2,\ldots,y^{\prime(0)}_N\sim \mathcal{N}(0,I)$
        \FOR{$k=1,2,\ldots,K$}
        \STATE  Obtain $y^{\prime(k)}_j,j=1,2,\ldots,N$ through Langevin step:
        \begin{equation*}
        y^{\prime(k)}_j=y^{\prime(k-1)}_j-\frac{\epsilon^2}{2}\nabla_y E_\theta(y^{\prime(k-1)}_j)+\epsilon z^{(k)}_j,
        \end{equation*}
        where $z^{(k)}_i\sim\mathcal{N}(0,I)$.
        \ENDFOR
        \STATE Compute the approximate gradient:
         \begin{equation*}
        \nabla_\theta \mathcal{L}(\theta)\!\approx \frac{1}{N} \sum_{i=1}^{N}\nabla_\theta E_\theta(y^{(K)}_i)-\frac{1}{N} \sum_{j=1}^{N}\nabla_\theta E_\theta(y^{\prime(K)}_j).
        \end{equation*}
        \STATE Update $\theta$ by gradient descent $\theta\leftarrow \theta-\eta \nabla_\theta \mathcal{L}(\theta)$
    \ENDWHILE
    \RETURN $\theta$
\end{algorithmic}
\label{alg}
\end{algorithm}
\vspace{-5pt}

\subsection{Relation with Classical Energy-Based Models}
The proposed EBRD model and the classical EBMs are formally very similar. 
However, a notable difference lies in the approach used to compute the gradients of the network parameter $\theta$.
The EBM computes the difference between the energy gradients of the original data and the generated samples. 
In contrast, the EBRD model computes the difference between the energy gradients of samples generated from the conditional distribution and those generated directly from the marginal distribution.
Fig. \ref{fig:comp} illustrates the sampling processes of the two models.

As $\beta\rightarrow+\infty$, the conditional distribution $p_\theta(y|x)$ converges to $\delta_x(y)$.
In this case, it can be assumed that the samples generated from the conditional distribution are identical to the original samples, suggesting that the classical EBM and the EBRD model become asymptotically equivalent.
\vspace{-4pt}
\begin{figure}[H]
    \centering
    \includegraphics[width=1\linewidth]{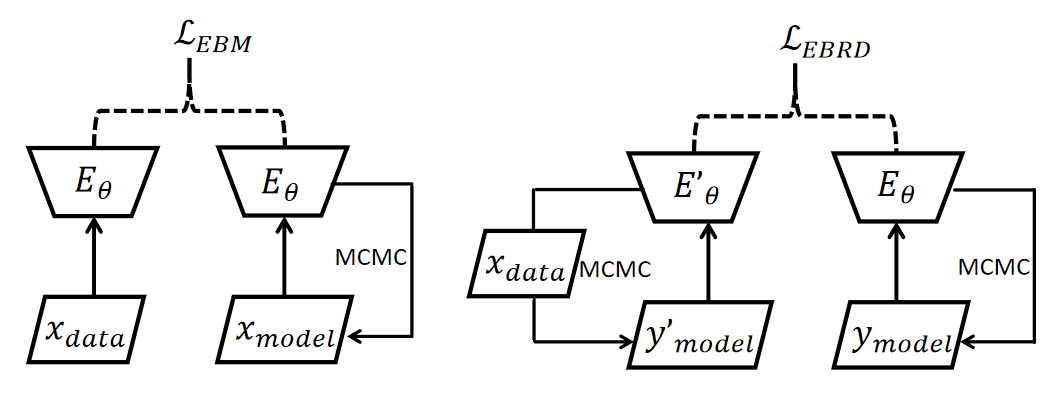}
    \caption{Comparison of sampling schemes between  the classical EBM and the EBRD model.}
    \label{fig:comp}
\end{figure}
\section{Experimental Results and Discussions}
\vspace{-2pt}
This section evaluates the effectiveness of the proposed EBRD model.
We first test the accuracy of estimating the RD function for scalar classical sources and vector Gaussian sources, and then visualize the reconstruction of conditional distributions on a two-dimensional Gaussian mixture model.
\subsection{Cases with Scalar Classical Sources}

This subsection estimates the RD function of two classical sources, i.e., the Gaussian source and the Laplacian source.

For the Gaussian source, we set $X\sim \mathcal{N}(0,1)$, and employ the squared Euclidean distance $\Vert x-y\Vert^2$ as the distortion measure. Its RD function is given by \cite{cover1999elements}:
\begin{equation*}
    R_{\text{gaussian}}(D)=
    \left\{
    \begin{aligned}
        &\frac{1}{2}\log(\frac{1}{D}),&0\leq D \leq 1\\
        &0,&D > 1
    \end{aligned}
    \right..
\end{equation*}

For the Laplacian source, the density function of 
$X$ is
\begin{equation*}
    p(x)=\frac{1}{2}\exp\left(-|x|\right).
\end{equation*}
We employ the $L_1$ distance $|x-y|$ as the distortion measure, and the correspoding RD function is given by \cite{berger2003rate}:
\begin{equation*}
    R_{\text{laplacian}}(D)=
    \left\{
    \begin{aligned}
        &\log(\frac{1}{D}),&0\leq D \leq 1\\
        &0,&D> 1
    \end{aligned}
    \right..
\end{equation*}

For the aforementioned two sources, we train the EBRD model with different values of $\beta$
 and compute the corresponding $(R, D)$ pairs.
For the Gaussian source, we set the number of steps $K=50$ and the step size $\epsilon=1.2\times10^{-2}$.
For the Laplacian source, we set  $K=80$ and  $\epsilon=3.5\times10^{-2}$.
Fig. \ref{fig:res1} compares the $(R,D)$ pairs given by the EBRD model with the corresponding theoretical  curves.
The experimental results demonstrate that our EBRD model  accurately computes the RD function, validating its effectiveness.
\vspace{-5pt}
\begin{figure}[H]
    \centering
    \includegraphics[width=0.49\linewidth]{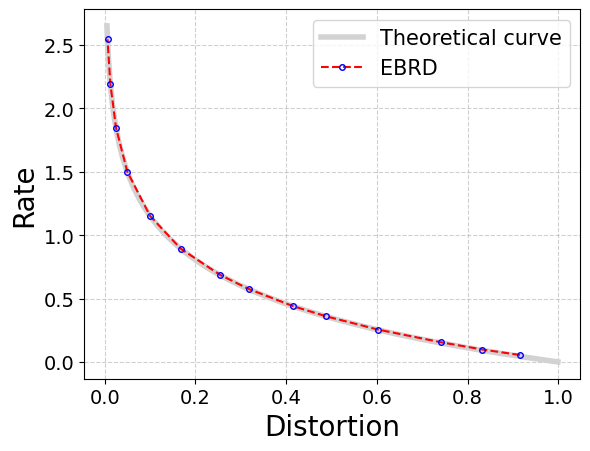}
    \includegraphics[width=0.48\linewidth]{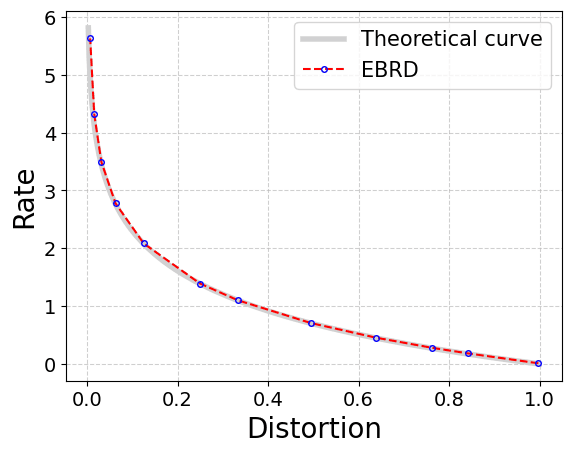}
    \caption{Comparison between the theoretical result and the EBRD algorithm on Gaussian (left) and Laplacian (right) sources.}
    \label{fig:res1} 
\end{figure}
\vspace{-5pt}
\subsection{Evaluation for Vector Gaussian Sources}
This subsection evaluates the performance of the proposed EBRD model on high-dimensional Gaussian sources.

Specifically, let $X$ follow the distribution $\mathcal{N}(0,\Sigma)$, where  $\Sigma$ admits the  decomposition $\Sigma=U\text{diag}(\sigma_1^2,\sigma_2^2,\ldots,\sigma_d^2)U^T$ with $U$ being an orthogonal matrix.
According to \cite{cover1999elements}, its RD function is
\begin{equation*}
    R(D)=
    \left\{
    \begin{aligned}
        &\sum_{i=1}^{d} \frac{1}{2}\log \frac{\sigma_i^2}{D_i},&0\leq D \leq \sum_{i=1}^{d} \sigma_i^2\\
        &0,&D>\sum_{i=1}^{d} \sigma_i^2
    \end{aligned}
    \right..
\end{equation*}
where
    $D_i\!=\!\min(\lambda,\sigma_i^2)$, 
and $\lambda$ is chosen to satisfy the equality $\sum_{i=1}^{d}D_i=D$.

In the experiment, we randomly $U$ and set $\sigma_i=2^{-\frac{i}{10}}$.
Fig. \ref{fig:res2} presents the experimental results obtained by the EBRD algorithm for dimensions $d=2,5,10$ compared with the NERD algorithm\cite{lei2022neural}.
It can be seen that both methods consistently obtain accurate results for $d=2$ and 5.
As for $d=10$, the results from both algorithms are accurate at low to medium rates.
However, for sufficiently high rates, errors emerge in both algorithms since the number of samples required for accurate computation grows exponentially with the rate, as explained in \cite{lei2022neural}.
\vspace{-6pt}
\begin{figure}[h]
    \centering
    \includegraphics[width=0.32\linewidth]{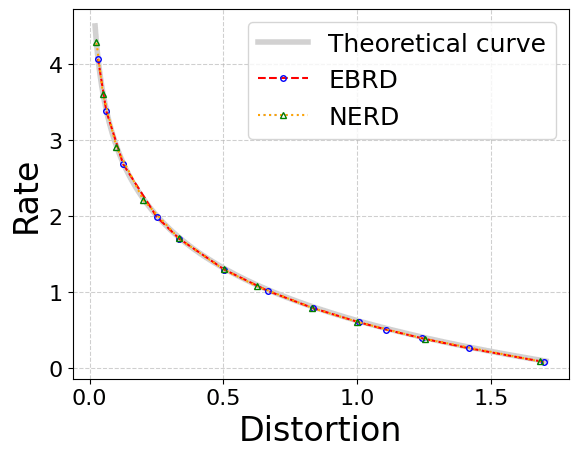}
    \includegraphics[width=0.32\linewidth]{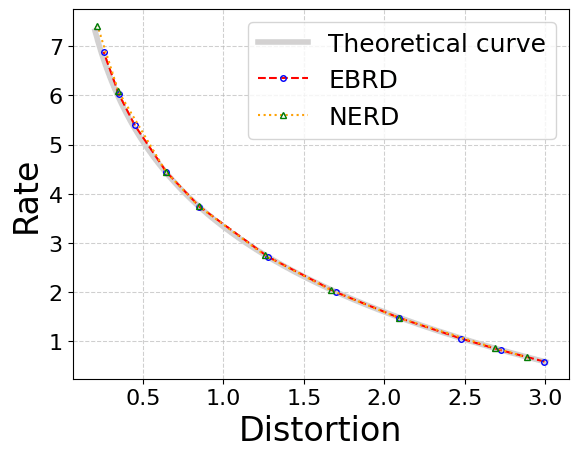}
    \includegraphics[width=0.32\linewidth]{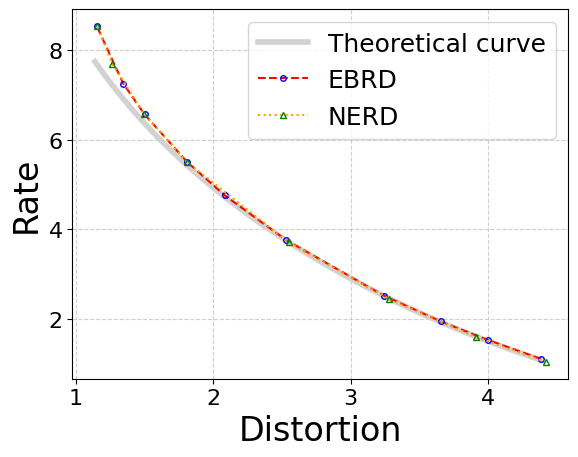}
    \caption{Comparison of EBRD and NERD algorithms with theoretical curves for vector Gaussian sources, with $d=2,5,10$ from left to right. }
    \label{fig:res2} 
\end{figure}
\vspace{-5pt}
\subsection{Visualization of Conditional Sampling}
This subsection adopts a two-dimensional Gaussian mixture model to visualize the optimal conditional distribution.

This distribution is constructed as a mixture of three two-dimensional Gaussian distributions with equal weights, each having distinct means and identical variances.
The density function of this distribution is
\vspace{-3pt}
\begin{equation*}
    p(x)=\sum_{i=1}^{3}\frac{1}{3}\mathcal{N}(\mu_i,\Sigma_i).
\end{equation*}
Specifically, the means of the three Gaussian components are placed equidistantly on a circle with a radius of 6. 
The variance of each component is set to $1$.
Samples are generated by first randomly selecting a component with equal probability and then drawing from the corresponding Gaussian distribution.

We sample 5,000 points from this distribution, visualizing the data distribution in Fig. \ref{fig:res3} (left) and presenting the RD curve obtained by the EBRD algorithm in Fig. \ref{fig:res3} (right).
Next, we select different $(R,D)$ pairs along the curve and employ the corresponding EBRD model to reconstruct the conditional distribution $p_\theta(y|x)$.
Fig. \ref{fig:res4} and Fig. \ref{fig:res5} illustrate the distribution of sample points $y$ conditioned on $x$ in various $(R,D)$ pairs along the curve.
It is noticed that with increasing rates,  the difference between the distribution of 
$y$ and that of the original source $x$ tends to diminish.
 Specifically, at $R = 5.96$ and $D = 0.02$, the two distributions become remarkably similar.
This observation is consistent with theoretical expectations.

\section{Conclusion}
In this work, we delve deeply into the connection between energy-based models and the RD problem, and propose an effective method to estimate the RD functions, i.e., the EBRD algorithm. 
The proposed algorithm is based on a reformulation of the RD problem by analogizing to the free energy in statistical physics, and then we derive a gradient descent algorithm similar to that in the classical EBMs.
%By analogizing to the free energy in statistical physics, we reformulate the RD problem and derive a contrastive gradient algorithm similar to the MLE model, which we named the EBRD algorithm.
%
The proposed EBRD algorithm not only accurately estimates the RD curve but also successfully reconstructs the optimal conditional distribution. 
Moreover, numerical experiments validate the effectiveness of the EBRD algorithm.
%future work：1、改进采样方案，提高对比梯度的稳定性 2、应用到更大规模的数据集。

As for the future work, we aim to improve the sampling scheme involved in the EBRD model to enhance the stability of computation.
This refinement would allow the algorithm to handle larger-scale datasets, further expanding its applicability and potential impact.

\begin{figure}[H]
    \centering
    \includegraphics[width=0.49\linewidth]{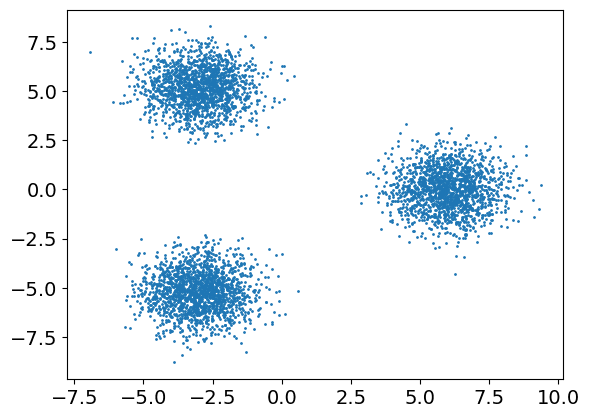}
    \includegraphics[width=0.47\linewidth]{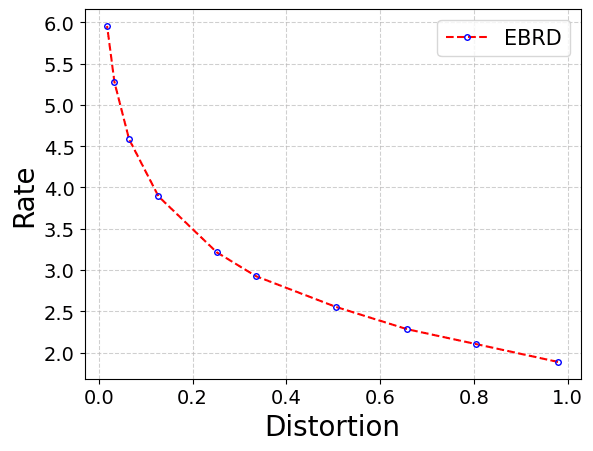}
    \caption{5000 sampled points from the Gaussian mixture model (left) and its RD curve (right).}
    \label{fig:res3} 
\end{figure}

\begin{figure}[H]
	\begin{minipage}{0.48\linewidth}
		\centerline{\includegraphics[width=\textwidth]{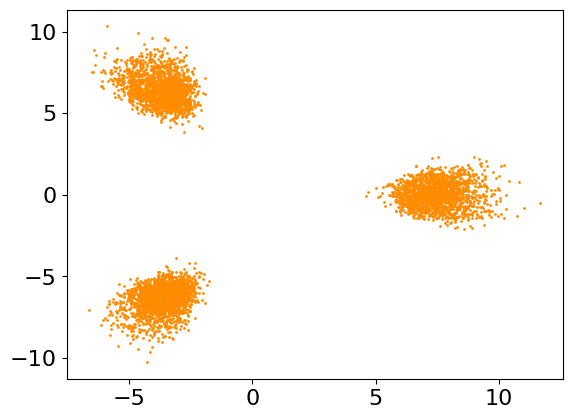}}
		\centerline{$R=1.89,D=0.98$}
	\end{minipage}
	\begin{minipage}{0.48\linewidth}
		\centerline{\includegraphics[width=\textwidth]{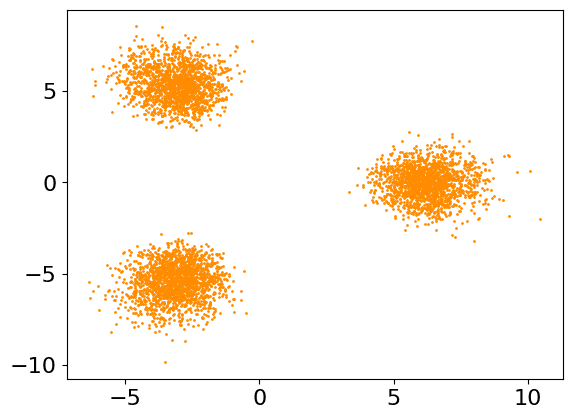}}
	 
		\centerline{$R=2.55,D=0.51$}
	\end{minipage}
 
	\caption{Visualization of conditional sampling with different $(R,D)$ pairs}
	\label{fig:res4}
\end{figure}

\begin{figure}[H]
	
	\centering
    \includegraphics[width=0.75\linewidth]{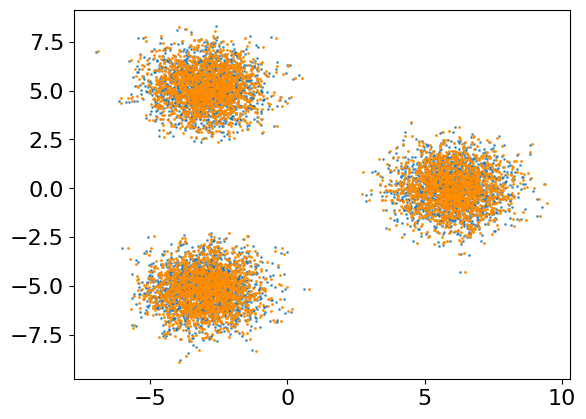}

	\caption{Comparison between the original distribution and the conditional sampling distribution at $R=5.96, D=0.02$. The blue scatter ones represent the original data, and the orange points represent the sampled data.}
	\label{fig:res5}
\end{figure}
% \section{acknowledgements}
% The authors would like to thank Prof. Hao Wu for his valuable discussions and helpful advice on optimization and information theory.

\bibliographystyle{bibliography/IEEEtran}
\bibliography{bibliography/ERD_REF}
%\bibliography{references}  

% \newpage 
\input{appendix.tex}

\end{document}

%% file: appendix.tex
\begin{appendices}

\section{Estimation of rate and distortion}
\label{app:A}
Similar to \cite{lei2022neural}, we generate samples from the distribution $Y$ using the well-trained energy function $E_\theta$, and subsequently estimate the corresponding $(R,D)$ pairs.
Specifically, we sample $x_1,x_2,\dots,x_N$ from $P_X$, and $y_1,y_2,\dots,y_N$ from $Q_Y^{\theta}$.
Then the objective function $\mathcal{L}(\theta)$ is approximated as
\begin{equation}
    \widehat{\mathcal{L}}(\theta)=-\frac{1}{N}\sum_{i=1}^{N}\log\left(\frac{1}{N}\sum_{j=1}^Ne^{-\beta \rho(x_i,y_j)}\right).
    \label{lhat}
\end{equation}

Further, the distortion is estimated by 
\begin{equation*}
    \begin{aligned}
        D &= \int_{\mathcal{X\times Y}}p(x)p(y|x)\rho(x,y) \mathrm{d}x\mathrm{d}y\\
        &=\int_{\mathcal{X}}p(x)\frac{\int_{\mathcal{Y}}q(y)e^{-\beta \rho(x,y)}\rho(x,y)\mathrm{d}y}{\int_{\mathcal{Y}}q(y)e^{-\beta \rho(x,y)}\mathrm{d}y}\mathrm{d}x\\
        &\approx \frac{1}{N}\sum_{i=1}^{N}\frac{\sum_{j=1}^{N}e^{-\beta \rho(x_i,y_j)}\rho(x_i,y_j)}{\sum_{j=1}^{N}e^{-\beta \rho(x_i,y_j)}}.
    \end{aligned}
\end{equation*}

Then, the rate is obtained by
\begin{equation*}
    R = \widehat{\mathcal{L}}(\theta)-\beta D.
\end{equation*}

Next, we analyze the asymptotic relationship between the estimated objective function \eqref{lhat} and the original objective function \eqref{vf}, which is precisely described by the following theorem.

\begin{theorem}
Suppose that there exists $M$, such that for $\forall x\in\mathcal{X}$, it satisfies that $\int_{\mathcal{Y}}e^{-\beta\rho(x,y)}\mathrm{d}y<M$. Let $F^*$ be the optimal value of \eqref{vf}, and then there exists a compact parameter domain $\Theta$ such that the optimal value of $\widehat{\mathcal{L}}(\theta)$, defined as
\begin{equation*}
    \widehat{F}_{N}=\inf_{\theta\in \Theta}\widehat{\mathcal{L}}(\theta),
\end{equation*}
 converges to $F^*$ with probability one.
Specifically, for $\forall \epsilon>0$, there exist a positive  integer $N_0$, such that $\forall N>N_0$,
\begin{equation}
    |\widehat{F}_{N}-F^*|<\varepsilon,\quad a.e.
\label{ae}
\end{equation}
where “a.e.” stands for almost everywhere.
\end{theorem}

\begin{proof}
    Fix $\varepsilon>0$. Let $q^*(y)$ be the optimal solution that minimizes \eqref{vf}, and it takes the form of the Boltzmann distribution:
    \begin{equation*}
        q^*(y)=\frac{e^{-E^*(y)}}{Z^*}.
    \end{equation*}
    where $Z^*=\int_\mathcal{Y}e^{-E^*(y)}\mathrm{d}y$ and $E^*(y)>0$.
    
    Fix $\xi>0$. By the universal approximation theorem   \cite{hornik1989multilayer,kidger2020universal}, there exist $\Theta$ and $\hat{\theta}\in\Theta$, such that $\forall y\in\mathcal{Y}$,
    $$
    |E^*(y)-E_{\hat{\theta}}(y)|<\xi.
    $$

    Then it holds that
    \begin{equation*}
    \begin{aligned}
        |e^{-E^*(y)}-e^{-E_{\hat{\theta}}(y)}|&=e^{-E^*(y)}|1-e^{E^*(y)-E_{\hat{\theta}}(y)}|\\
        &<e^{-E^*(y)}(e^\xi-1).
    \end{aligned}
    \end{equation*}

    From this, one obtains
    \begin{equation*}
    \begin{aligned}
        &|q^*(y)-q^{\hat{\theta}}(y)|=|\frac{e^{-E^*(y)}}{Z^*}-\frac{e^{-E_{\hat{\theta}}(y)}}{Z_{\hat{\theta}}}|\\
        \leq&\frac{e^{-E^*(y)}|Z_{\hat{\theta}}-Z^*|+Z^*|e^{-E^*(y)}-e^{-E_{\hat{\theta}}(y)}|}{Z^*Z_{\hat{\theta}}}\\
        <&\frac{2(e^\xi-1)}{(2-e^\xi)Z^*}.
    \end{aligned}
    \end{equation*}
    % \noindent Note that $e^x$ is Lipschitz continuous on the interval $(-\infty,0]$ with constant 1, we have
    % $$
    %     |e^{\beta\sum_yp(y| x)\log r^*(y| z)}\!-\!e^{\beta\sum_yp(y| x)\log r_{\hat{\theta}}(y| z)}|
    %     <\beta\eta.
    % $$
    % %    
    % Then we obtain
    % {\footnotesize
    % \begin{equation*}
    %     \left|E_{P_Z}\!\!\left(e^{\beta\sum_yp(y| x)\log r^*(y| Z)}\right)\!\!-\!\!E_{P_Z}\!\!\left(e^{\beta\sum_yp(y| x)\log r_{\hat{\theta}}(y| Z)}\right)\right|<\beta\eta.
    % \end{equation*}
    % }
    
    % \noindent Since $\log$ is a continuous function, there exists a $\eta$, such that the above formula yields 
    This further implies that
    \begin{equation*}
    \begin{aligned}
        &\left|\int_\mathcal{Y}q^*(y)e^{-\beta \rho(x,y)}\mathrm{d}y-\int_\mathcal{Y}q_{\hat{\theta}}(y)e^{-\beta \rho(x,y)}\mathrm{d}y\right|\\<&M\frac{2(e^\xi-1)}{(2-e^\xi)Z^*}.
    \end{aligned}
    \end{equation*}
    
    Since the right-hand side of the above inequality tends to zero as $\xi \rightarrow 0$  and $\log x$ is continuous, there exists a positive number $\xi$, such that the following relation holds:
    \begin{equation*}
        \left|\log\int_\mathcal{Y}q^*(y)e^{-\beta \rho(x,y)}\mathrm{d}y-\log\int_\mathcal{Y}q_{\hat{\theta}}(y)e^{-\beta \rho(x,y)}\mathrm{d}y\right|<\frac{\varepsilon}{2}.
    \end{equation*}

    This means that by choosing a suitable $\hat{\theta}$, one has $|F^*-\mathcal{L}(\hat{\theta})|<\frac{\varepsilon}{2}$.
    Next, define
    \begin{equation*}
    F_{\Theta}=\inf_{\theta\in\Theta}\mathcal{L}(\hat{\theta}).
    \end{equation*}
    From this definition, one obtains $F^*\leq F_{\Theta} \leq \mathcal{L}(\hat{\theta})$, and then
    \begin{equation}
        |F^*-F_{\Theta}|<\frac{\varepsilon}{2}.
        \label{ie1}
    \end{equation}
    %
    % We then write
    % $$\begin{aligned}
    % |G_{\Theta}-\widehat{G}_{mn}|
    % =&|\inf_{\theta\in\Theta}F(\theta)-\inf_{\theta\in\Theta}G_\theta|\\
    % \leq&\sup_{\theta\in\Theta}|F(\theta)-G_\theta|.
    % \end{aligned}
    % $$
    
    Further, one has 
    $$
    |\widehat{F}_{N}-F_{\Theta}|
    =|\inf_{\theta\in\Theta}\widehat{\mathcal{L}}(\theta)-\inf_{\theta\in\Theta}\mathcal{L}(\theta)|
    \leq\sup_{\theta\in\Theta}|\widehat{\mathcal{L}}(\theta)-\mathcal{L}(\theta)|.$$
    It follows from the law of large numbers \cite{gray2009probability} that for any given $\epsilon>0$, there exists $N_0$, such that 
    $\forall N>N_0$ and with probability one, $|\widehat{\mathcal{L}}(\theta)-\mathcal{L}(\theta)|\leq\frac{\varepsilon}{2}$.
    Therefore
    \begin{equation}
        |\widehat{F}_{N}-F_{\Theta}|<\frac{\varepsilon}{2}.
        \label{ie2}
    \end{equation}
    Combining \eqref{ie1} and \eqref{ie2}, one finally has
    \begin{equation*}
        |F^*-\widehat{F}_{N}|<\varepsilon,\,\forall\,N>N_0.
    \end{equation*}
    This concludes the proof of the theorem. 
    %Here, we have completed the proof of the theorem.
\end{proof}

\end{appendices}

%% file: main.bbl
\begin{thebibliography}{10}
\providecommand{\url}[1]{#1}
\csname url@samestyle\endcsname
\providecommand{\newblock}{\relax}
\providecommand{\bibinfo}[2]{#2}
\providecommand{\BIBentrySTDinterwordspacing}{\spaceskip=0pt\relax}
\providecommand{\BIBentryALTinterwordstretchfactor}{4}
\providecommand{\BIBentryALTinterwordspacing}{\spaceskip=\fontdimen2\font plus
\BIBentryALTinterwordstretchfactor\fontdimen3\font minus
  \fontdimen4\font\relax}
\providecommand{\BIBforeignlanguage}[2]{{%
\expandafter\ifx\csname l@#1\endcsname\relax
\typeout{** WARNING: IEEEtran.bst: No hyphenation pattern has been}%
\typeout{** loaded for the language `#1'. Using the pattern for}%
\typeout{** the default language instead.}%
\else
\language=\csname l@#1\endcsname
\fi
#2}}
\providecommand{\BIBdecl}{\relax}
\BIBdecl

\bibitem{shannon1948mathematical}
C.~E. Shannon, ``A mathematical theory of communication,'' \emph{The Bell
  System Technical Journal}, vol.~27, no.~3, pp. 379--423, 1948.

\bibitem{shannon1959coding}
C.~E. Shannon \emph{et~al.}, ``Coding theorems for a discrete source with a
  fidelity criterion,'' \emph{IRE Nat. Conv. Rec}, vol.~4, no. 142-163, p.~1,
  1959.

\bibitem{balle2017end}
J.~Ball{\'e}, V.~Laparra, and E.~P. Simoncelli, ``End-to-end optimized image
  compression,'' in \emph{5th International Conference on Learning
  Representations, ICLR 2017}, 2017.

\bibitem{balle2020nonlinear}
J.~Ball{\'e}, P.~A. Chou, D.~Minnen, S.~Singh, N.~Johnston, E.~Agustsson, S.~J.
  Hwang, and G.~Toderici, ``Nonlinear transform coding,'' \emph{IEEE Journal of
  Selected Topics in Signal Processing}, vol.~15, no.~2, pp. 339--353, 2020.

\bibitem{blahut1972computation}
R.~Blahut, ``Computation of channel capacity and rate-distortion functions,''
  \emph{IEEE transactions on Information Theory}, vol.~18, no.~4, pp. 460--473,
  1972.

\bibitem{10619450}
L.~Chen, S.~Wu, W.~Zhang, H.~Wu, and H.~Wu, ``On convergence of discrete
  schemes for computing the rate-distortion function of continuous source,'' in
  \emph{2024 IEEE International Symposium on Information Theory (ISIT)}, 2024,
  pp. 410--415.

\bibitem{lei2022neural}
E.~Lei, H.~Hassani, and S.~S. Bidokhti, ``Neural estimation of the
  rate-distortion function with applications to operational source coding,''
  \emph{IEEE Journal on Selected Areas in Information Theory}, vol.~3, no.~4,
  pp. 674--686, 2022.

\bibitem{yang2022towards}
Y.~Yang and S.~Mandt, ``Towards empirical sandwich bounds on the
  rate-distortion function,'' in \emph{International Conference on Learning
  Representations}, 2022.

\bibitem{higgins2017beta}
I.~Higgins, L.~Matthey, A.~Pal, C.~P. Burgess, X.~Glorot, M.~M. Botvinick,
  S.~Mohamed, and A.~Lerchner, ``beta-vae: Learning basic visual concepts with
  a constrained variational framework.'' \emph{ICLR (Poster)}, vol.~3, 2017.

\bibitem{yang2024estimating}
Y.~Yang, S.~Eckstein, M.~Nutz, and S.~Mandt, ``Estimating the rate-distortion
  function by wasserstein gradient descent,'' \emph{Advances in Neural
  Information Processing Systems}, vol.~36, 2024.

\bibitem{ho2020denoising}
J.~Ho, A.~Jain, and P.~Abbeel, ``Denoising diffusion probabilistic models,''
  \emph{Advances in neural information processing systems}, vol.~33, pp.
  6840--6851, 2020.

\bibitem{DBLP:conf/iclr/0011SKKEP21}
Y.~Song, J.~Sohl{-}Dickstein, D.~P. Kingma, A.~Kumar, S.~Ermon, and B.~Poole,
  ``Score-based generative modeling through stochastic differential
  equations,'' in \emph{9th International Conference on Learning
  Representations, {ICLR} 2021, Virtual Event, Austria, May 3-7, 2021}, 2021.

\bibitem{li2024rate}
Q.~Li and C.~Guyot, ``Rate-distortion theory by and for energy-based models,''
  \emph{IEEE Transactions on Communications}, 2024.

\bibitem{DBLP:conf/icml/BelghaziBROBHC18}
M.~I. Belghazi, A.~Baratin, S.~Rajeswar, S.~Ozair, Y.~Bengio, R.~D. Hjelm, and
  A.~C. Courville, ``Mutual information neural estimation,'' in
  \emph{Proceedings of the 35th International Conference on Machine Learning,
  {ICML} 2018, Stockholmsm{\"{a}}ssan, Stockholm, Sweden, July 10-15, 2018},
  ser. Proceedings of Machine Learning Research, 2018, pp. 530--539.

\bibitem{landau1980statistical}
L.~D. Landau and E.~M. Lifshitz, \emph{Statistical Physics, Part I},
  3rd~ed.\hskip 1em plus 0.5em minus 0.4em\relax Oxford: Pergamon, 1980.

\bibitem{parisi1981correlation}
G.~Parisi, ``Correlation functions and computer simulations,'' \emph{Nuclear
  Physics B}, vol. 180, no.~3, pp. 378--384, 1981.

\bibitem{grenander1994representations}
U.~Grenander and M.~I. Miller, ``Representations of knowledge in complex
  systems,'' \emph{Journal of the Royal Statistical Society: Series B
  (Methodological)}, vol.~56, no.~4, pp. 549--581, 1994.

\bibitem{lecun2006tutorial}
Y.~LeCun, S.~Chopra, R.~Hadsell, M.~Ranzato, F.~Huang \emph{et~al.}, ``A
  tutorial on energy-based learning,'' \emph{Predicting structured data},
  vol.~1, no.~0, 2006.

\bibitem{song2021train}
Y.~Song and D.~P. Kingma, ``How to train your energy-based models,''
  \emph{arXiv preprint arXiv:2101.03288}, 2021.

\bibitem{kumar2019maximum}
R.~Kumar, S.~Ozair, A.~Goyal, A.~Courville, and Y.~Bengio, ``Maximum entropy
  generators for energy-based models,'' \emph{arXiv preprint arXiv:1901.08508},
  2019.

\bibitem{rose1994mapping}
K.~Rose, ``A mapping approach to rate-distortion computation and analysis,''
  \emph{IEEE Transactions on Information Theory}, vol.~40, no.~6, pp.
  1939--1952, 1994.

\bibitem{cover1999elements}
T.~M. Cover, \emph{Elements of Information Theory}.\hskip 1em plus 0.5em minus
  0.4em\relax John Wiley \& Sons, 1999.

\bibitem{berger2003rate}
T.~Berger, ``Rate-distortion theory,'' \emph{Wiley Encyclopedia of
  Telecommunications}, 2003.

\bibitem{hornik1989multilayer}
K.~Hornik, M.~Stinchcombe, and H.~White, ``Multilayer feedforward networks are
  universal approximators,'' \emph{Neural Networks}, vol.~2, no.~5, pp.
  359--366, 1989.

\bibitem{kidger2020universal}
P.~Kidger and T.~Lyons, ``Universal approximation with deep narrow networks,''
  in \emph{Conference on Learning Theory (COLT)}, 2020, pp. 2306--2327.

\bibitem{gray2009probability}
R.~M. Gray, \emph{Probability, Random Processes, and Ergodic Properties}.\hskip
  1em plus 0.5em minus 0.4em\relax Springer Science \& Business Media, 2009.

\end{thebibliography}
